\documentclass[aps,prl,reprint,longbibliography]{revtex4-1}

\usepackage{graphicx}
\usepackage{amsmath}
\usepackage{amssymb}
\usepackage{amsthm}
\usepackage{algpseudocode}
\usepackage{algorithm}
\usepackage{enumitem}
\usepackage{float}
\usepackage{hyperref}
\usepackage{tikz}
\usepackage{thm-restate}

\newcommand{\be}{\begin{equation}}
\newcommand{\ee}{\end{equation}}
\newcommand{\ba}{\begin{array}}
\newcommand{\ea}{\end{array}}
\newcommand{\bea}{\begin{eqnarray}}
\newcommand{\eea}{\end{eqnarray}}

\newcommand{\ra}{\rangle}
\newcommand{\la}{\langle}

\newcommand{\calE}{{\cal E }}

\newcommand{\calS}{{\cal S }}

\newcommand{\calZ}{{\cal Z }}
\newcommand{\FF}{\mathbb{F}}

\newcommand{\CC}{\mathbb{C}}

\newcommand{\RR}{\mathbb{R}}

\newcommand{\ket}[1]{|#1\rangle}
\newcommand{\bra}[1]{\langle #1|}
\newcommand{\braket}[2]{\langle #1|#2\rangle}

\newtheorem{dfn}{Definition}

\newtheorem{lemma}{Lemma}

\newtheorem{theorem}{Theorem}

\newtheorem{problem}{Problem}

\begin{document}
\title{How to simulate quantum measurement without computing marginals}
\author{Sergey Bravyi$^1$}
\author{David Gosset$^{2,3,4}$}
\author{Yinchen Liu$^{2,3}$}
\affiliation{$^1$ IBM Quantum, IBM T.J. Watson Research Center}
\affiliation{$^2$ Department of Combinatorics and Optimization, University of Waterloo}
\affiliation{$^3$ Institute for Quantum Computing, University of Waterloo}
\affiliation{$^4$ Perimeter Institute for Theoretical Physics, Waterloo}
\begin{abstract}
We describe and analyze algorithms for classically simulating measurement of an $n$-qubit quantum state $\psi$ in the standard basis, that is, sampling a bit string $x$
from the probability distribution $|\la x|\psi\ra|^2$. Our algorithms reduce the sampling task to
computing poly$(n)$ amplitudes of $n$-qubit states;
unlike previously known techniques they do not require computation of marginal probabilities.
First we consider the case where $|\psi\rangle=U|0^n\rangle$ is the output state of an $m$-gate quantum circuit $U$. We propose an exact sampling algorithm
which involves computing $O(m)$ amplitudes 
of $n$-qubit states generated by subcircuits of $U$
spanned by the first $t=1,2,\ldots,m$ gates. We show that our algorithm can significantly accelerate
quantum circuit simulations based on
tensor network contraction methods or low-rank stabilizer decompositions. 
As another striking consequence we obtain an efficient classical simulation algorithm for measurement-based quantum computation with the surface code resource state on any planar graph, generalizing a previous algorithm which was known to be efficient only under restrictive topological constraints on the ordering of single-qubit measurements. Second, we consider the case in which $\psi$ is the unique ground state of a local Hamiltonian with a spectral gap that is lower bounded by an inverse polynomial function of $n$. We prove that a simple Metropolis-Hastings Markov Chain mixes rapidly to the desired probability distribution provided that $\psi$ obeys a certain technical condition,  which we show is satisfied for all sign-problem free Hamiltonians. This gives a sampling algorithm which involves computing  $\mathrm{poly}(n)$ amplitudes of $\psi$. 
\end{abstract}

\maketitle

There is strong evidence that quantum circuits cannot be simulated efficiently using a classical computer. Likewise, physical properties of locally interacting quantum many-body systems are unlikely to be classically accessible in the general case. Nevertheless, classical simulation techniques are widely used in quantum computation and condensed matter physics. To some extent this is out of necessity, as a means to go beyond the limits of pen-and-paper calculation. But it is also facilitated by the fact that mathematicians, computer scientists, and physicists have identified certain remarkable quantum systems where efficient classical simulation is possible. These include the family of Clifford circuits (simulable using the stabilizer formalism \cite{gottesman1997stabilizer}), systems that are equivalent to noninteracting fermionic particles including matchgate circuits \cite{matchgates,terhal2002classical} and the 2D Ising model \cite{onsager,barahona1982computational,kasteleyn1961statistics}  (via fermionic linear optics \cite{flo}), gapped 1D quantum many-body systems \cite{1darea, landau2015polynomial} or shallow quantum circuits in a 1D geometry \cite{jozsa2006simulation} (tensor network methods \cite{perez2006matrix, tensornetworks, markov2008simulating}), and ferromagnetic spin systems \cite{jerrum1993polynomial, bravyi2015monte,bravyi2017polynomial} (Markov chain Monte Carlo methods). Such examples are rare and insightful; each provides a glimpse of a facet of the quantum-classical boundary and informs our understanding of hard-to-simulate quantum resources.
Perhaps more importantly, the above algorithmic techniques can often be extended to more general settings with an increased computational cost. For example, the classical simulation algorithms based on low-rank stabilizer decompositions \cite{bravyi2016trading, bravyi2016improved, bravyi2019simulation} have a runtime which scales exponentially only in the number of non-Clifford gates in a quantum circuit. Tensor-network based simulation methods for quantum circuits \cite{markov2008simulating} have a runtime which scales exponentially only in the treewidth of a graph which describes the connectivity of the circuit. A large body of recent work (see, e.g., \cite{pednault2019leveraging, huang2020classical, pan2020contracting, gray2021hyper, pan2021simulating, pan2021solving, closinggap})  has focused on optimizing practical implementations of tensor network methods for the benchmark task of sampling from the output distribution of random quantum circuits,  in response to the quantum experiment \cite{arute2019quantum}. We expect classical simulation will continue to be a key technique for validation and verification of near-term quantum devices, and in the study of quantum many-body systems.

\begin{figure*}
\begin{minipage}{0.45\linewidth}
\begin{algorithm}[H]
	\caption{Qubit-by-qubit sampling\label{qbqsampling}}
\hspace*{\algorithmicindent} \hspace{-87pt}\textbf{Input:}  An $n$-qubit quantum state $\psi$.\\
\hspace*{\algorithmicindent} \hspace{-53pt} \textbf{Output:} $x\in \{0,1\}^n$ with probability $|\langle x|\psi\rangle|^2$.
	\begin{algorithmic}[1]
\State{Sample $x_1\in \{0,1\}$ from the probability distribution $\pi_1(x_1)$.}
			\For{$j=2$ to $n$}
				\State{Sample $x_j\in \{0,1\}$ from the probability distribution $ \pi_j(x_{1}\ldots x_{j-1}x_j)/\pi_{j-1}(x_{1}\ldots x_{j-1})$.}
			      \EndFor
		\State{\textbf{return} $x=x_1x_2\ldots x_n$}
		\end{algorithmic}
\end{algorithm}
\end{minipage}
\hfill
\begin{minipage}{0.45\linewidth}
\begin{algorithm}[H]
	\caption{Gate-by-gate sampling\label{sampling}}
	\hspace*{\algorithmicindent}  \hspace{-20pt}\textbf{Input:}  An $n$-qubit quantum circuit $U=U_m\cdots U_2 U_1$.\\
 \hspace*{\algorithmicindent} \hspace{-40pt} \textbf{Output:} $x\in \{0,1\}^n$ with probability $|\langle x|U|0^n\rangle|^2$. 
	\begin{algorithmic}[1]
	\State{$x\gets 0^n$}
			\For{$t=1$ to $m$}
		\State{$A\gets \{1,2,\ldots,n\}\setminus \mathrm{supp}(U_{t})$ 
		}
			\State{$S\gets \{y\in \{0,1\}^n \, : \, y_A=x_A\}$}
		\State{Sample $x\in S$ from the probability distribution $P_t(x)/\sum_{y\in S} P_t(y)$}
			      \EndFor
		\State{\textbf{return} $x$}
		\end{algorithmic}
\end{algorithm}
\end{minipage}
\end{figure*}

In this work we provide new techniques for a fundamental and ubiquitous task: simulating measurement of a quantum state $\psi$ in the standard basis. Throughout we shall assume $\psi$ is a normalized $n$-qubit quantum state and so our goal is to sample from the output distribution $|\langle x|\psi\rangle|^2$, where $x\in \{0,1\}^n$.

 It is well known that this task can be performed given the ability to compute any marginal probability of the form 
\begin{equation}
\pi_j(y)\equiv \langle \psi|\big( |y\rangle\langle y|\otimes I_{n-j}\big)|\psi\rangle \qquad y\in \{0,1\}^j.
\label{eq:marg}
\end{equation}
The standard \textit{qubit-by-qubit} sampling algorithm uses the chain rule for conditional probabilities to 
simulate measurement of each qubit $j=1,2,\ldots, n$ in sequence. It samples each measurement outcome $x_j\in \{0,1\}$ for $j=1,2,\ldots, n$ from its conditional distribution given the values of all previously sampled bits. This qubit-by-qubit algorithm---
stated formally as Algorithm~\ref{qbqsampling} below---
is the usual way to reduce the task of weak simulation (our sampling task) to strong simulation (computing a given probability or marginal). It is applicable in a wide variety of contexts as it works for any quantum state $\psi$. It has been deployed in countless works.

The runtime of the qubit-by-qubit algorithm is determined by the cost of computing the marginal probabilities $\pi_1(x_1), \pi_2(x_1x_2),\ldots, \pi_n(x_1x_2\ldots x_n)$. In particular, the total runtime is at most $n$ times the maximum runtime of computing a marginal of the form Eq.~\eqref{eq:marg}. The latter runtime may vary widely depending on the method used, and whether or not the state $\psi$ has special structure that can be exploited. In the cases we consider in this work (see below), computing marginals is \#P-hard in the worst case and it is expected that any algorithm must scale exponentially with the number of qubits.

Here we describe alternatives to the qubit-by-qubit algorithm, for two important families of quantum states $\psi$: 
output states of polynomial-size quantum circuits and 
unique ground states of local Hamiltonians with inverse polynomial spectral gap. In other words we give alternative efficient reductions from weak to strong simulation for these families of states. 
Our reductions differ from the qubit-by-qubit algorithm in that they do not require computation of marginal probabilities. 
Instead, our 
algorithms make a polynomial number of calls to a subroutine that computes \textit{amplitudes} of $n$-qubit states. We describe settings in which our new reductions provide vast improvements in total runtime for the task of simulating measurement.

\section{Simulation of quantum circuits}

Consider the task of 
sampling a bit string from 
the output distribution of a quantum circuit $U$
with $m$ gates
such that each gate is a unitary operator acting non-trivially on at most $k$ qubits.
We show how to reduce the sampling task to the one of computing
amplitudes of subcircuits of $U$ spanned by the first $t$ gates where $t=1,2,\ldots,m$.
The total number of amplitudes that one needs
to compute is at most $m2^k$.

To fix notation, suppose $U=U_m\cdots U_2 U_1$ is a quantum circuit acting on $n$ qubits.
Each gate $U_i$ acts non-trivially on a subset of qubits
 $\mathrm{supp}(U_i)\subseteq [n]$ called the support of $U_i$. Here and below $[n]\equiv \{1,2,\ldots,n\}$.
Let 
\be
P_t(x)=|\la x|U_t\cdots U_2 U_1|0^n\ra|^2
\ee
be the output distribution generated by the first $t$ gates of $U$
and $P_0(x) = |\la x|0^n\ra|^2=\delta_{x,0^n}$. Given a subset of qubits $A\subseteq [n]$
and a bit string $x\in \{0,1\}^n$, let $x_A\in \{0,1\}^{|A|}$
be the restriction of $x$ onto $A$.

Consider
the \textit{gate-by-gate} sampling algorithm 
described above as Algorithm~\ref{sampling}.
We claim that this algorithm outputs a bit string $x$ sampled from the desired distribution $P_m(x)=|\la x|U|0^n\ra|^2$.
Indeed,  let $Q_t(x)$ be the  probability distribution of $x$ at the 
end of the $t$-th iteration of the {\bf for} loop. Let 
$Q_0(x)=P_0(x)=\delta_{x,0^n}$.
Suppose we have already proved that $Q_{t-1}(x)=P_{t-1}(x)$ for all $x$.
Consider the $t$-th iteration of the {\bf for} loop and let $x$ be the bit
string sampled at the previous iteration. 
Let $P_t(x_A):=\sum_{y\, : \, y_A=x_A} P_t(y)$ be the marginal probability of
$x_A$ with respect to $P_t$.
Note that $\sum_{y\in S} P_t(y) = P_t(x_A)$. Thus
\begin{align*}
Q_{t}(y) &= \sum_{x\, : \, x_A = y_A} Q_{t-1}(x) \frac{P_{t}(y)}{P_{t}(y_A)}
=\sum_{x\, : \, x_A = y_A} P_{t-1}(x) \frac{P_{t}(y)}{P_{t}(y_A)}\\
&=\frac{P_{t-1}(y_A) P_{t}(y)}{P_{t}(y_A)}
= P_{t}(y).
\end{align*}
To get the last equality note that $U_{t}$ acts trivially 
on  $A$ which implies $P_{t-1}(y_A)=P_t(y_A)$ since $U_{t}$ is unitary.
Thus $Q_t(x)=P_t(x)$ for all $t$ and $x$.

To execute
line~5 one needs to compute $P_t(y)$ for each $y\in S$.
Since $|S|\le 2^k$, overall one needs to compute at most $m2^k$ output probabilities
$P_t(y)$ with $t=1,\ldots,m$ 
In the special case of CNOT+SU(2) circuits one needs to use lines~3-5 only 
if $U_t$ is a single-qubit gate. If $U_t$ is a CNOT,
replace lines~3-5 by  $|x\ra\gets U_t|x\ra$.
Then effectively $k=1$ and one needs to compute at most $2m$ output probabilities.
Likewise, if $U_t$ is a diagonal gate such as a  $Z$-rotation or CZ, one can skip
the $t$-th iteration of the $\bf for$ loop since $P_t(x)=P_{t-1}(x)$.

We note that Algorithm~\ref{sampling} can be applied almost verbatim to 
the task of sampling the output distribution of  an {\em adaptive} quantum circuit
which includes intermediate measurements such that each gate
may be classically controlled by outcomes of all previous measurements \footnote{
Since Algorithm~\ref{sampling} measures every qubit 
after applying each gate, no modification are needed to simulate
intermediate measurements.
Let us agree that 
once a qubit has been measured, all subsequent gates act trivially on this qubit. 
Then the  dependence of a gate $U_t$ on the outcomes of the earlier measurements
can be modeled by allowing $U_t$ to be classically controlled by 
the bit string $x_A$, where the register $A$ is defined at line~4.
Otherwise the algorithm and its analysis remains unchanged.}.

We now discuss situations in which the gate-by-gate algorithm may be preferable to the qubit-by-qubit algorithm. 

Let $f(n,d)$ be the cost of computing an amplitude of an $n$-qubit circuit with depth $d$ using some strong simulation method, such as tensor network contraction~\cite{markov2008simulating}.
We would expect a marginal probability such as
$\langle 0^n| U^{\dagger} (|y\rangle \langle y|\otimes I )U |0^n\rangle$
to have a cost comparable to $f(n,2d)$, since in general our best upper bound on the depth of the operator appearing in the expectation value is $2d+1$. Thus we expect the gate-by-gate algorithm to have a significant advantage over the qubit-by-qubit algorithm whenever $f(n,2d)/f(n,d)$ is large.  It may be helpful to consider two extreme cases. If we use the Schr\"{o}dinger simulation method which stores the entire $n$-qubit state in memory as a complex vector of length $2^n$ and then applies gates using sparse matrix-vector multiplication,  then we have $f(n,d)=O(nd2^n)$ and the advantage is only a constant factor. On the other hand, if we use a simple method that only requires $\mathrm{poly}(n,d)$ memory---the Feynman sum-over-paths technique---then $f(n,d)$ scales exponentially in $d$ and the advantage is substantial. The best polynomial-space algorithm we are aware of has a runtime scaling as $f(n,d)=O(n\cdot (2d)^{n+1})$ \cite{aaronson2016complexity} and in this case the advantage of the gate-by-gate method is exponential in $n$. From these examples we expect the gate-by-gate algorithm to be advantageous in memory-limited classical simulations where the entire state-vector cannot be loaded into classical memory.

In practice, tensor-network simulators may use heuristic algorithms to optimize their space and memory usage. To test whether or not our method can provide an advantage when using such methods, we used CoTenGra \cite{gray2021hyper} to optimize and estimate the tensor-network contraction costs of sampling once from the output distribution of a $49$-qubit-depth-$16$ 2D quantum circuit using both algorithms. To impose memory constraints, we used CoTenGra's slicing feature to restrict the maximum size of intermediate tensors. From TABLE \ref{Table:cotengra}, we observe that the gate-by-gate algorithm incurs significantly less slicing overheads, agreeing with the intuitive arguments. It is important to note that the contraction costs are estimated without actually performing the contractions. For more details, see the Supplementary Material.

\begin{table}[t]
    \begin{tabular}{c|cccc}
\multicolumn{1}{}{} & \multicolumn{4}{c}{$\log_2$ of max intermediate tensor size}\\
& 29 & 31 & 33 & 35\\
\hline
$\log_2(F_G)$ & 58.4433 &     58.3197 &    58.1232 &    58.1339  \\
$\log_2(F_Q)$ & 75.4501 &     73.1768 &    71.0325 &    68.9512 \\
$F_Q/F_G$ & 131690 &  29677 &  7693 &  1804  \\
\end{tabular}
\caption{FLOP count comparison for memory-limited tensor network simulation of a 2D depth-$16$ circuit on a $7\times 7$ grid of qubits.  Here $F_G, F_Q$ are, respectively, the FLOP counts of the gate-by-gate and qubit-by-qubit algorithms.}
\label{Table:cotengra}
\end{table}

The gate-by-gate algorithm can also provide runtime improvements for simulation methods based on low-rank stabilizer decompositions. Recall that a stabilizer state of $n$ qubits is a state of the form $C|0^n\ra$, where $C$ is a Clifford circuit composed of CNOT, Hadamard, and $S=\mathrm{diag}(1,i)$ gates.
The (exact) stabilizer rank $\chi(\alpha)$ of a quantum state $\alpha$ is the minimum integer $r$ such that $\alpha$ can be expressed as a linear combination of $r$ stabilizer states with complex coefficients~\cite{bravyi2016trading}.
As a simple example, suppose 
$|\psi\ra=U|0^n\ra$, where
$U$ is a circuit of size $poly(n)$ composed of Clifford gates and at most $\ell$ single-qubit gates $T=\mathrm{diag}(1,e^{i\pi/4})$. In this case 
it was shown~\cite{bravyi2016improved,qassim2021improved} that 
$\chi(\psi)\leq \chi(T^{\otimes \ell})
\le O(2^{0.3963\ell})$
where $|T\rangle\sim |0\rangle+e^{i\pi/4}|1\rangle$ is the single-qubit magic state \cite{bravyi2016improved}. 
It is known that any
amplitude of $n$-qubit stabilizer state
can be computed (including the overall phase)
in time $poly(n)$~\cite{bravyi2019simulation},
see also~
\footnote{More precisely, one can compute any amplitude of an $n$-qubit stabilizer state $\phi$ in time $O(n^2)$ provided that $\phi$ is specified by the so-called CH-form~\cite{bravyi2019simulation}. Computing the CH-form of $\phi$ starting from a Clifford circuit $C$ that prepares it, starting from the all-zeros computational basis state,  takes time $O(cn^2)$ where $c$ is the number of gates in $C$. }.
Since $\psi$ is a linear combination of $\chi(\psi)$ stabilizer states, any amplitude of $\psi$ can be computed
in time $poly(n)\chi(\psi)$. It follows that 
the gate-by-gate algorithm can sample the output distribution of $U$ 
in time $poly(n) \chi(T^{\otimes \ell})$.
The previous best known algorithm for this task 
based on the qubit-by-qubit simulation strategy
had a runtime that scales quadratically with $\chi(T^{\otimes \ell})$ \cite{bravyi2016trading}. There is strong evidence that this quantity increases exponentially with $\ell$ \cite{bravyi2016trading, huang2020explicit}, in which case we improve the \textit{exponent} of the runtime for the task of exact sampling. The fastest sampling algorithms based on stabilizer-rank methods, such as the sum-over-Cliffords method \cite{bravyi2019simulation}, allow for some small error in total variation distance from the true output distribution. In the Supplemental Material we show how the gate-by-gate algorithm can also be used to improve the runtime of such methods. While this is a more practical setting, the improvement is less dramatic as it only concerns polynomial prefactors in the runtime.

Our final example involves
a measurement-based quantum computation (MBQC)~\cite{raussendorf2003measurement}.
Recall that MBQC with an $n$-qubit resource state $\phi$
involves a sequence of $n$ single-qubit measurements
performed on a state 
\[
|\psi\ra = (U_1\otimes U_2 \otimes \cdots \otimes U_n)|\phi\ra,
\]
where $U_j$ are arbitrary single-qubit unitary operators.
Each unitary $U_j$ may depend on the outcomes of all previous measurements, according to some efficiently computable rule. 
For example, measurement of $\phi$ in the Fourier basis defined
by the Quantum Fourier Transform can be implemented by MBQC with the resource state $\phi$~\cite{griffiths1996semiclassical}.
MBQC is equivalent to the standard circuit-based quantum computation if one chooses $\phi$ as the 2D cluster state~\cite{raussendorf2003measurement}.
Here we choose $\phi$ as the Kitaev's surface code state~\cite{bravyi1998quantum} 
on a planar graph $G$, e.g. the 2D square lattice. 
It is known~\cite{bravyi2007measurement} that any amplitude of $\psi$ can be 
computed in time $O(n^3)$ by expressing it as the partition
function of the Ising model on the dual graph $G^*$
and using the seminal result by Barahona~\cite{barahona1982computational}.
This implies that the gate-by-gate
algorithm can efficiently simulate 
MBQC with the surface code state on any planar graph
for any temporal order of measurements. 
To the best of our knowledge, this is the first efficient classical algorithm for this task. 
A previous method~\cite{bravyi2007measurement}, based on the qubit-by-qubit sampling paradigm, provides an efficient simulation of such MBQC  only under certain restrictive topological constraints on the temporal order of measurements~\footnote{This algorithm can efficiently simulate any MBQC with the surface code state
such that the subset of qubits measured at every time step
spans a connected subgraph of $G$. The same should hold for the subset of unmeasured qubits.}. 
Moreover, in the Supplemental Material  we prove that computing certain marginal probabilities of $\psi$
required for the qubit-by-qubit algorithm is a $\#P$-hard problem. This suggests that this algorithm is incapable of efficiently simulating MBQC with the surface code state
for an arbitrary order of measurements.

 So far we have assumed that the output probabilities $P_t(x)$ can be computed exactly.
However, numerical simulations are exact only to within machine precision. 
Furthermore, some algorithms for  simulation of quantum circuits~\cite{bravyi2019simulation, vidal2003efficient,pan2020contracting}  are only meant to \textit{approximate} output probabilities. 
This leads to the question of whether Algorithm~\ref{sampling} is robust against errors in approximating the probabilities $P_t(x)$. Suppose that a subroutine is available for exactly computing amplitudes
of some $n$-qubit states $|\phi_t\ra$ such that  $\| |\phi_t\rangle - U_t\ldots U_2U_1|0^n\rangle\| \le \epsilon_t$ for all $t=1,2,\ldots,m$. Define probability distributions
\[
R_t(x) =  |\la x|U_t|\phi_{t-1}\ra|^2\|\phi_{t-1}\|^{-2},
\]
where $t=1,2,\ldots,m$.
Consider a modified version of Algorithm \ref{sampling} in which we replace $P_t$ by $R_t$ in line 5. In the Supplemental Material we prove the following lemma.
\begin{lemma}[\bf Robustness to errors]
\label{lemma:robustness}
Let $Q$ be the probability distribution describing the output of a modified version of Algorithm \ref{sampling} in which the approximation $R_t$ is used in place of $P_t$ in line 5. Then 
\be
\|Q - P_m\|_1:=\sum_{x\in \{0,1\}^n} |Q(x) - P_m(x)| \le 16\sum_{t=1}^{m-1} \epsilon_t.
\ee
\end{lemma}
 
\section{Simulation of ground states}
Suppose $\psi$ is the unique ground state
of a Hamiltonian $H$ describing a system of spins or fermions
 with few-body interactions.
Applying such Hamiltonian $H$ to any basis vector can flip only $O(1)$ bits~\footnote{Here we assume that the 
fermionic Hamiltonians are mapped to qubits using the second quantization method
and the Jordan Wigner transformation. Fock basis vectors are identified with $n$-bit strings, where $n$ is the number of fermionic modes.}.
More formally, let $d(x,y)$ be the Hamming distance between bit strings $x,y\in \{0,1\}^n$.
We require that 
\be
\label{H1}
\la x|H|y\ra=0 \quad \mbox{unless $d(x,y)\le k$}
\ee
for some fixed locality parameter $k=O(1)$. 
Equivalently, the expansion of $H$ in the Pauli basis
can only include products of single-qubit Pauli operators $X,Y,Z$
with at most $k$ factors $X$ and $Y$.
Let $\gamma>0$ be the spectral gap of $H$ separating the ground energy from the rest of the spectrum. 

As before, our goal is to sample a bit string $x\in \{0,1\}^n$ from 
the  distribution $\pi(x)=|\la x|\psi\ra|^2$
given  a subroutine for computing amplitudes of $\psi$.
More precisely, we shall only need a subroutine for computing
the ratio $\pi(y)/\pi(x)$ for given strings $x,y$.
The sampling algorithm takes as input an initial string $x_{in}$ such that $\pi(x_{in})$ is non-negligible
and outputs a sample from a distribution $\epsilon$-close to $\pi$ in the total variation distance.
The number of calls to the amplitude computation subroutine
scales as 
\begin{equation}
\label{sampling_runtime}
T \sim \frac{ n^k s}{\gamma} \log{\left( \frac1{\pi(x_{in}) \epsilon} \right)},
\end{equation}
where 
$s$ is a sensitivity parameter 
quantifying how
much the amplitude of $\psi$ can 
change upon flipping a few bits of $x$.
More formally, 
\begin{equation}
\label{sensitivity}
s= \max_{x\ne y} \; \frac{|\la y|H|x\ra\la x|\psi\ra |}{|\la y|\psi\ra|},
\end{equation}
where $x,y\in \{0,1\}^n$ and the maximization only includes strings $y$ such that $\la y|\psi\ra\ne 0$.
We can prove a general upper bound on $s$ only 
when $H$ is a sign-problem-free Hamiltonian,
a.k.a. stoquastic~\cite{bravyi2006complexity}.
Such Hamiltonians are defined by the property that 
all off-diagonal matrix elements of $H$ in the standard basis
are real and non-positive.
In the Supplementary Material we prove that
$s\le \max_x \la x|H|x\ra - E_0$ for any stoquastic Hamiltonian $H$
with the ground energy $E_0$. We leave as an open question whether the runtime dependence on $s$
can be avoided.

Our sampling algorithm  is the standard 
Metropolis-Hastings Markov Chain Monte Carlo method.
This method is often used in practice for  simulating measurement of approximations to quantum ground states, e.g. those based on neural networks \cite{carleo2017solving}.
It is often used as a heuristic 
 even if rigorous
bounds on the mixing time of the Markov chain are unavailable. 
Here we  prove that the  Metropolis-Hastings  Markov chain is rapidly mixing 
whenever the inverse spectral gap
$1/\gamma$ and the sensitivity parameter $s$ scale at most
polynomially with $n$.
In the Supplementary Material we describe a family of 
Hamiltonians $H$ for which the required amplitude computation subroutine can be implemented efficiently
and the sensitivity parameter obeys $s\le poly(n)$. This family includes some non-stoquastic Hamiltonians.

Before proceeding, we note that the sampling task considered here can be viewed as a generalization of the quantum circuit sampling task from the previous section. Indeed, one obtains an alternative algorithm for the latter by combining the method of this Section with the Feynman-Kitaev circuit-to-Hamiltonian mapping \cite{kitaev2002classical}.  The resulting algorithm for quantum circuit sampling has similar features
but is arguably less elegant and has less favorable runtime than the gate-by-gate sampling algorithm.
We also note that our algorithm is different from the Quantum Monte Carlo method
which applies only to sign-problem-free Hamiltonians $H$.
In contrast, our sampling algorithm can be applied to some Hamiltonians
with the sign problem, see the Supplementary Material for details.

Let $\calS\subseteq \{0,1\}^n$ be the support of $\pi$ such that $x\in \calS$ iff $\pi(x)>0$.
Define a Metropolis-Hastings type  Markov chain  with the state space $\calS$ 
such that the probability to transition from $x\in \calS$ to $y\in \calS\setminus \{x\}$ in one step is 
\begin{equation}
\label{Pxy}
P_{xy}=\frac12 Q_{xy} \cdot \min\{1, \pi(y)/\pi(x)\},
\end{equation}
where $Q$ is a symmetric proposal distribution that, given $x\in \calS$, selects a new binary string $y$ by choosing a uniformly random subset of bits of size at most $k$ and flipping them. More formally,
\begin{equation}
\label{proposal}
Q_{xy}=\begin{cases}\frac{1}{N} & \mbox{if $d(x,y)\leq k$}\\ 0 & \text{otherwise},\end{cases}
\end{equation}
where $N\equiv \sum_{j=0}^{k}\binom{n}{j}=O(n^k)$. 

One can directly check that this Markov chain is irreducible, aperiodic, and 
satisfies detailed balance with respect to the distribution $\pi$, i.e., $\pi(x)P_{xy}=\pi(y)P_{yx}$. 
Therefore $\pi$ is the unique limiting distribution of $P$. Let $\pi^{(t)}(x)$ be a distribution
obtained by performing $t$ steps of the Markov chain $P$ starting from a fixed state $x_{in}\in \calS$.
Using Proposition~3 of Ref.~\cite{diaconis1991geometric} one gets
\begin{equation}
\label{TVDbound}
\| \pi^{(t)} - \pi\|_1 \le \frac{\lambda_1^t}{2\sqrt{\pi(x_{in})}},
\end{equation}
where $\lambda_1\in [0,1)$ is the second largest eigenvalue of $P$
(here we noted that all eigenvalues of $P$ are non-negative since $P_{xx}\ge 1/2$).
Thus $\| \pi^{(t)} - \pi\|_1 \le \epsilon$ as long as $t\ge T$, where
\begin{equation}
\label{T}
T = \frac{\log{(2 \epsilon \sqrt{\pi(x_{in})})}}{\log{(\lambda_1)}}
\end{equation}
can be viewed as the mixing time.
A well-known variational characterization of the second
largest eigenvalue of a Markov chain~\cite{diaconis1991geometric} gives
\begin{equation}
\label{variational}
1-\lambda_1 = \inf_\phi \frac{\calE(\phi,\phi)}{\mathrm{Var}(\phi)}
\end{equation}
where the infimum is taken over all non-constant functions $\phi \, : \, \calS \to \RR$,
\[
\calE(\phi,\phi) = \frac12 \sum_{x,y\in \calS} \pi(x) P_{xy} (\phi(x)-\phi(y))^2
\]
is the so-called Dirichlet form, and $\mathrm{Var}(\phi)$ is the variance of $\phi(x)$
with respect to $\pi$. Let $\phi\, : \, \calS\to \RR$ be a function that 
achieves the infimum in Eq.~(\ref{variational}) such that 
\begin{equation}
\label{variational1}
(1-\lambda_1)  \mathrm{Var}(\phi) = \calE(\phi,\phi).
\end{equation}
Define an (unnormalized) $n$-qubit state 
\[
|\psi^\perp\ra = \sum_{x\in \calS} (\phi(x)-\mu) \psi(x)|x\ra,
\]
where $\psi(x)\equiv \la x|\psi\ra$ and $\mu=\sum_{x\in \calS} \pi(x) \phi(x)$ is the mean value
of $\phi(x)$. One can easily check that $\la \psi|\psi^\perp\ra=0$ and $\|\psi^\perp\|^2=\mathrm{Var}(\phi)$.
Let $E_0$ be the ground energy of $H$ and $H'=H-E_0I$.
Then $H'|\psi\ra=0$ and the second smallest eigenvalue of $H'$ is $\gamma$.
It follows that $\psi^\perp$ has energy at least $\gamma$ with respect to $H'$,  that is,
\begin{equation}
\label{variational1'}
\gamma \|\psi^\perp\|^2 \le \la \psi^\perp|H'|\psi^\perp\ra.
\end{equation}
Using the identity $H'|\psi\ra=0$
one gets
\begin{align*}
\la \psi^\perp|H'|\psi^\perp\ra &=
\sum_{x,y\in \calS} \psi^*(x) \psi(y) \phi(x) \phi(y) \la x|H'|y\ra \\
& = -\frac12\sum_{x,y\in \calS} \psi^*(x) \psi(y) (\phi(x)- \phi(y))^2 \la x|H'|y\ra, \\
& \le \frac12\sum_{x,y\in \calS}
(\phi(x)- \phi(y))^2 |
\psi(x) \psi(y) \la x|H|y\ra|.
\end{align*}
Here we noted that $\la y|H'|x\ra=\la y|H|x\ra$ for $x\ne y$.
From Eq.~(\ref{sensitivity}) one gets
\[
|\psi(x) \psi(y) \la x|H|y\ra|\le s \cdot \min{(\pi(x),\pi(y))}
\]
for any $x\ne y$.
Combining this bound and  Eqs.~(\ref{H1},\ref{Pxy},\ref{proposal})
one gets 
\[
|\psi(x) \psi(y) \la x|H|y\ra|\le 2Ns\pi(x) P_{xy}
\]
for any $x\ne y$.
We conclude that 
\begin{equation}
\label{variational2}
\gamma \|\psi^\perp\|^2 \le 2Ns \calE(\phi,\phi).
\end{equation}
Combining Eqs.~(\ref{variational1},\ref{variational2}) and the identity 
$\| \psi^\perp\|^2 = \mathrm{Var}(\phi)$ one arrives at 
$1-\lambda_1 \ge \frac{\gamma}{2Ns}$. Now
the runtime scaling claimed in Eq.~(\ref{sampling_runtime})
follows from Eq.~(\ref{T}).

We conclude by noting that similar arguments were used by Crosson and Bowen~\cite{crosson2017quantum}
to establish isoperimetric inequalities for probability distributions
associated with ground states of gapped local Hamiltonians. 
Such inequalities can be used to bound the conductance~\cite{levin2017markov} of the Markov chain 
considered above,
which provides an upper bound on its mixing time.
However we expect this bound to scale as $O(1/\gamma^2)$ which is quadratically worse
compared with Eq.~(\ref{sampling_runtime}).

\paragraph{Acknowledgments}
We thank Ramis Movassagh for helpful discussions and for the suggestion to use
CoTenGra library. We also thank Johnnie Gray for suggesting the use of the dynamic slicing option in CoTenGra. SB is supported in part by the IBM Research Frontiers Institute.
DG and YL acknowledge the support of the Natural Sciences and Engineering Research Council of Canada through grant number RGPIN-2019-04198. DG also acknowledges the support of the Canadian Institute for Advanced Research, and IBM Research. Research at Perimeter Institute is supported in part by the Government of Canada through the Department of Innovation, Science and Economic Development
Canada and by the Province of Ontario through the Ministry of Colleges and Universities.
\bibliographystyle{unsrt}
\bibliography{mybib}
\appendix
\onecolumngrid

\section{Robustness of the gate-by-gate algorithm}

\begin{proof}[Proof of Lemma \ref{lemma:robustness}]
Define states 
\[
|\psi_t\ra = U_t \cdots U_2 U_1|0^n\ra.
\]
By assumption, we have a subroutine for computing amplitudes
of states $\phi_t$ such that 
\be
\| \psi_t - \phi_t\| \le \epsilon_t
\ee
for all $t\ge 1$. We set $|\phi_0\ra=|\psi_0\ra=|0^n\ra$.
A simple algebra then gives
\be
\label{approx_eq1}
\| \psi_t - \frac{\phi_t}{\|\phi_t\|} \| \le 2\epsilon_t.
\ee
In this section we consider a modified version of the gate-by-gate algorithm that uses the probability distribution 
\[
R_t(x)=\frac{|\la x|U_t|\phi_{t-1}\ra|^2}{\|\phi_{t-1}\|^2}
\]
for $t\ge 1$
in place of $P_t$ in line 5.
Let $Q_t(x)$ be the  probability distribution of $x$ at the 
end of the $t$-th iteration of the {\bf for} loop of the modified algorithm.  It suffices to show that 
\be
\label{inductionPQ}
\| P_t - Q_t\|_1 \le \delta_t, \qquad \delta_t:=16\sum_{s=1}^{t-1} \epsilon_s 
\ee
for all $t=1,2,\ldots,m$. Here it is understood that $\delta_1=0$.
We shall use induction in $t$.
The base of induction is $t=1$.  In this case $|\phi_0\ra = |\psi_0\ra$ and the 
analysis performed in the main text shows that 
$Q_1=P_1$. Thus $\delta_1=0$ proving the base of induction.

Consider now the induction step. Define distributions
\[
\Phi_t(x) = \frac{|\la x|\phi_t\ra|^2}{\|\phi_t\|^2},
\]
where $t=1,2,\ldots,m$.
The $t$-th iteration of the {\bf for} loop can be described by a stochastic matrix $M_t$
of size $2^n\times 2^n$ such that $Q_t=M_t Q_{t-1}$. Here we consider probability
distributions as column vectors.
Repeating the same arguments as in the main text one gets
$R_t=M_t \Phi_{t-1}$. Thus 
\[
\| P_t - Q_t\|_1 \le \| P_t - R_t\|_1 + \| R_t - Q_t\|_1 =
 \| P_t - R_t\|_1 + \| M_t(\Phi_{t-1} - Q_{t-1})\|_1
\le  \| P_t - R_t\|_1 + \|\Phi_{t-1} - Q_{t-1}\|_1.
\]
Here we used the triangle inequality and
noted that the multiplication by a stochastic matrix
does not increase the total variation distance between distributions.
Applying the triangle inequality to the last term gives
\be
\label{M2}
\| P_t - Q_t\|_1 \le  \| P_t - R_t\|_1 +  \|\Phi_{t-1} - P_{t-1}\|_1 + 
\| P_{t-1} - Q_{t-1}\|_1.
\ee
By the induction hypothesis, $\|P_{t-1} - Q_{t-1}\|_1\le \delta_{t-1}$. Thus
\be
\label{M3}
\| P_t - Q_t\|_1 \le  \| P_t - R_t\|_1 +  \|\Phi_{t-1} - P_{t-1}\|_1 + 
\delta_{t-1}.
\ee
For any normalized quantum states $|\psi\ra$ and $|\phi\ra$
let $\Psi(x)=|\la x|\psi\ra|^2$ and $\Phi(x)=|\la x|\phi\ra|^2$ be the corresponding
distributions. We have 
\be
\label{L1L2norms}
\| \Psi - \Phi\|_1 = 2 \max_{S \subseteq \{0,1\}^n}\; |\Psi(S) - \Phi(S) |
= 2 \max_{S \subseteq \{0,1\}^n}\; | \la \psi |\Pi_S |\psi\ra - \la \phi|\Pi_S |\phi\ra|
\le 4\| \psi - \phi\|,
\ee 
where $\Pi_S :=\sum_{x\in S} |x\ra\la x|$. From Eqs.~(\ref{approx_eq1},\ref{L1L2norms}) one gets
\be
\label{M4}
\| \Phi_{t-1} - P_{t-1}\|_1 \le 4 \| \frac{\phi_{t-1}}{\|\phi_{t-1}\|} - \psi_{t-1}\| \le 8\epsilon_{t-1}.
\ee 
By definition, the distributions $P_t$ and $R_t$ are obtained by performing a measurement
on the states $|\psi_t\ra=U_{t}|\psi_{t-1}\ra$ and $U_t |\phi_{t-1}\ra/\|\phi_{t-1}\|$ respectively.
From Eqs.~(\ref{approx_eq1},\ref{L1L2norms}) one gets
\be
\label{M5}
\| P_t - R_t\|_1
\le 4\| U_t |\psi_t\ra - \frac{U_t |\phi_{t-1}\ra}{\| \phi_{t-1}\|} \|
=4\|  \psi_t  - \frac{\phi_{t-1}}{\| \phi_{t-1}\|} \|\le 8 \epsilon_{t-1}.
\ee
Combining Eqs.~(\ref{M3},\ref{M4},\ref{M5}) gives
\[
\| P_t - Q_t\|_1 \le \delta_{t-1} + 16\epsilon_{t-1}  = \delta_t
\]
completing  the induction step. 
\end{proof}

\subsection{Application to the sum-over-Cliffords simulator}
\label{sec:Clifford}

Suppose each state $|\psi_t\ra = U_t \cdots U_2U_1|0^n\ra$
can be approximated by a state 
\be
\label{stab1}
|\phi_t\ra = \sum_{\alpha=1}^{\chi_t} c_{t,\alpha} |\omega_{t,\alpha}\ra,
\ee
where $c_{t,\alpha}$ are complex coefficients and $|\omega_{t,\alpha}\ra$
are $n$-qubit stabilizer states.
Below we assume that 
\be
\label{stab2}
\| \psi_t - \phi_t\| \le \epsilon_t
\ee
for all $t$. The desired approximation can be computed
using the sum-over-Cliffords decomposition of Ref.~\cite{bravyi2019simulation}.
As was shown in that work, 
the number of terms in the sum Eq.~(\ref{stab1}) known 
as the stabilizer rank grows exponentially with the number
of non-Clifford gates among $U_1,U_2,\ldots,U_t$, namely, 
\be
\chi_t  = \frac1{\epsilon_t^2} \prod_{s=1}^t \xi(U_s),
\ee
where $\xi(U_t)\ge 1$ with the equality iff $U_t$ is a Clifford gate.
For example, if $U_t$ is a  single-qubit $Z$-rotation $e^{-i(\theta/2)Z}$ with $\theta \in [0,\pi/2]$ then
\[
\xi(U_t) = \left( \cos{(\theta/2)} + \tan{(\pi/8)} \sin{(\theta/2)} \right)^2.
\]
For details, see Ref.~\cite{bravyi2019simulation}. Below we assume that 
\be
\label{stab3}
\sum_{t=1}^{m-1} \epsilon_t =\delta/16.
\ee
Then Lemma~\ref{lemma:robustness} implies that the
modified version of Algorithm~\ref{sampling}
samples the output distribution of $U$ within a statistical error $\delta$.
Following Ref.~\cite{bravyi2019simulation}, we shall
assume that each stabilizer state in the decomposition Eq.~(\ref{stab1}) is specified by
the so-called  CH-form~\cite{bravyi2019simulation}
which is a data structure for describing stabilizer states
including the overall phase. 
As shown in~\cite{bravyi2019simulation},
any amplitude of an $n$-qubit stabilizer state
specified by the CH-form can be computed in time $O(n^2)$. It follows that any amplitude
of the state $|\phi_t\ra$ can be computed in time $O(\chi_t n^2)$.
Thus the total cost (runtime) of the modified Algorithm~\ref{sampling} is at most 
\be
\label{stab4}
C = O(n^2) \sum_{t=1}^{m-1} \chi_t = O(n^2) \sum_{t=1}^{m-1} \frac{\xi(U_1)\cdots \xi(U_t)}{\epsilon_t^2}.
\ee
Minimizing the cost $C$ over the variables $\epsilon_1,\ldots,\epsilon_{m-1}\ge 0$ subject to the
constraint Eq.~(\ref{stab3}) gives
\be
\label{stab5} 
C = \frac{O(n^2)}{\delta^2} \left( \sum_{t=1}^{m-1} \prod_{s=1}^t [\xi(U_s)]^{1/3} \right)^3,
\ee
with the optimal choice of $\epsilon_t$ being 
\be
\epsilon_t = \frac{\delta}{16} \cdot \frac{\eta_t^{1/3}}{\sum_{s=1}^{m-1} \eta_s^{1/3}}, \qquad
\eta_t : = \prod_{s=1}^{t} \xi(U_s).
\ee
Since $\xi(U_t)\ge 1$ for all $t$ and $\xi(U_t)>1$ for non-Clifford gates, one should expect that the
sum over $t$ in Eq.~(\ref{stab5}) is dominated by the last few terms. Then 
\be
\label{stab6}
C = \frac{O(n^2)}{\delta^2} \prod_{t=1}^{m-1} \xi(U_t).
\ee
For comparison, sampling the output distribution of $U$ using the methods
of Ref.~\cite{bravyi2019simulation} would have cost (runtime) at least
$O(n^6 \chi_m)=O(n^6/\delta^2)\prod_{t=1}^{m} \xi(U_t)$.

Following  Ref.~\cite{bravyi2019simulation},
the above discussion ignores the cost of computing the CH-form
of stabilizer states in the decomposition Eq.~(\ref{stab1}).
This cost does not depend on the number of samples that one has to generate since the CH-form needs to be computed only once. As shown in~\cite{bravyi2019simulation}, 
one can compute the CH-form of an $n$-qubit stabilizer state specified by a Clifford circuit with $c$ gates
in time $O(cn^2)$.

\section{Quantum circuit simulation numerical details}
\begin{table}[h]
    \begin{tabular}{c|cccc}
\multicolumn{1}{}{} & \multicolumn{4}{c}{$\log_2$ of max intermediate tensor size}\\
& 29 & 31 & 33 & 35\\
\hline
$\log_2(F_G)$ & 58.4984 &     58.3358 &    58.0783 &    58.1070   \\
$\log_2(F_Q)$ & 75.3782 &     73.1171 &    71.1410 &    68.9029 \\
$F_Q/F_G$ & 120590 &  28160 &  8556 &  1778   \\
\end{tabular}
    \begin{tabular}{c|cccc}
\multicolumn{1}{}{} & \multicolumn{4}{c}{$\log_2$ of max intermediate tensor size}\\
& 29 & 31 & 33 & 35\\
\hline
$\log_2(F_G)$ & 58.3999 &     58.2307 &    58.1988 &    58.1414 \\
$\log_2(F_Q)$ & 75.5524 &     73.1137 &    71.0621 &    68.9284 \\
$F_Q/F_G$ & 145691 &  30217 &  7451 &  1767 \\
\end{tabular}
\caption{FLOP count comparison for memory-limited tensor network simulation of a 2D depth-$16$ circuit on a $7\times 7$ grid of qubits.  Here $F_G, F_Q$ are, respectively, the FLOP counts of the gate-by-gate and qubit-by-qubit algorithms; the variability is present but small and the FLOP counts are stable across independent runs with identical parameter settings.}
\label{Table:cotengra_multiple_runs}
\end{table}

In this section, we describe the setup used to obtain \autoref{Table:cotengra} as well as \autoref{Table:cotengra_multiple_runs}. Experiments with both the qubit-by-qubit and gate-by-gate algorithms are facilitated through quimb \cite{gray2018quimb} and CoTenGra \cite{gray2021hyper}. We use CoTenGra's \verb|cotengra.ReusableHyperOptimizer| to perform contraction tree optimization for both algorithms using identical settings \verb|minimize=`flops'|, \verb|methods=[`kahypar']|, and \verb|max_repeats=512|. The maximum intermediate tensor size constraints are enforced by setting \verb|slicing_reconf_opts={`target_size':2**s}| for $s=29, 31, 33, 35$. The $49$-qubit-depth-$16$ circuit consists of a $7\times 7$ grid of qubits with alternating horizontal and vertical Haar random $2$-qubit gate patterns repeating every four layers. Although the numerics is performed on a random quantum circuit, the results should be representative for arbitrary circuits with this architecture since we do not expect the tensor network optimizer to take advantage of the random gate entries. Due to the internal randomness in CoTenGra's optimization algorithms, the estimated FLOP counts vary slightly between runs. In \autoref{Table:cotengra_multiple_runs}, we report the results from two other runs using the settings described above. The inverse relationship between $F_Q/F_G$ and the maximum intermediate tensor size is consistently observed. We mention that with sixty CPU cores, it takes about three days to execute one optimization run involving both the gate-by-gate and qubit-by-qubit algorithms. Adopting normal slicing (by setting \verb|slicing_opts|) as opposed to dynamic slicing (by setting \verb|slicing_reconf_opts|) reduces the optimization time from days to hours at the cost of significantly higher FLOP counts for both algorithms for all four choices of the \verb|target_size|.

To estimate the total contraction cost of the qubit-by-qubit algorithm, we adopt the implementation provided by quimb with the rehearse option (\verb|quimb.tensor.circuit.Circuit.sample_rehearse|). For an $n$-qubit circuit, the qubit-by-qubit algorithm needs to perform $n$ contractions, and the total cost $F_Q$ is computed by summing the cost of each contraction extracted from the optimizer \verb|opt| using \verb|opt.get_tree().contraction_cost()|.


For the gate-by-gate algorithm, we use \verb|quimb.tensor.circuit.Circuit.amplitude_rehearse| to estimate the cost of computing the amplitudes needed to execute line 5 of Algorithm~\ref{sampling}. The cost of each contraction is extracted from the optimizer \verb|opt| using \verb|opt.get_tree().contraction_cost()|. The sample $x\in S$ drawn at line 5 of Algorithm~\ref{sampling} is always assumed to be $0^n$ for an $n$-qubit circuit throughout the algorithm. The one-sample contraction cost $F_G$ is computed by summing the cost of computing the amplitudes in each iteration.

\section{Measurement-based computation with the surface code states}
\label{app:surface}

We begin by formally defining the surface code state
and the measurement-based quantum computation (MBQC). 
Let $G=(V,E)$ be a planar graph with $n$ edges.
We shall label edges by integers such that $E=\{1,2,\ldots,n\}$.
Suppose $x\subseteq E$ is a subset of edges. 
We  identify  $x$ with an $n$-bit string
which has '1' at the $j$-th position iff $x$ contains the $j$-th edge of $G$.
Let us say that $x$ is a {\em cycle} if
each vertex of $G$ has even number of incident edges from $x$.
Let $\calZ(G)$ be the set of all cycles in $G$.
Note that $\calZ(G)$ is a linear subspace of $\FF_2^n$.
One can choose a basis of $\calZ(G)$ such that each basis vector is the boundary of some face of $G$.
Place a qubit at every edge of $G$ and define an $n$-qubit state
$\psi_G$ which is the uniform superposition of all cycles,
\be
\label{surface_code}
|\psi_G\ra = \frac1{\sqrt{|\calZ(G)|}} \sum_{x\in \calZ(G)} \; |x\ra.
\ee
One can easily check that $|\psi_G\ra$ coincides with the Kitaev's surface code state~\cite{kitaev2003fault,bravyi1998quantum}
in the special case when $G$ is the 2D square lattice with open boundary conditions. 
Accordingly, we shall refer to  $\psi_G$ as the surface code state on the graph $G$.
By definition, 
measurement of $\psi_G$ returns a random uniformly distributed cycle $x\in \calZ(G)$.
Such cycle can be represented as $x=\sum_{j=1}^f r_j b_j {\pmod 2}$, where $f$ is the number of faces in the graph, $b_j\subseteq E$ is the boundary of the $j$-th face, and $r\in \{0,1\}^f$ is picked uniformly at random.
Since $f\le n$, one concludes that measurement of $\psi_G$ can be simulated in time $O(n)$.
Here we assumed that the specification of $G$ includes a list of all faces. 

From now on we consider MBQC with the resource state $\psi_G$. It is defined as a sequence of $n$ single-qubit measurements performed on a state 
\be
\label{surface_code_rotated}
|\psi\ra = (U_1 \otimes U_2 \otimes \cdots \otimes U_n)|\psi_G\ra,
\ee
where $U_j$ are arbitrary single-qubit unitary operators.
Qubits are measured sequentially in the order $1,2,\ldots,n$.
Let $x_i\in \{0,1\}$ be the measurement outcome on the $i$-th qubit.
The unitary $U_j$ can be chosen as a function of all previous measurement outcomes, that is, $U_j=U_j(x_1,x_2,\ldots,x_{j-1})$.
This function must be computable
in time $poly(n)$. As before, our goal is to sample a bit string $x\in \{0,1\}^n$ describing the outcomes of all $n$ measurements. 
Since the resource state $\psi_G$ is fixed throughout this section,
below we use the term MBQC without specifying the resource state.

Let us first discuss how to simulate MBQC using the  gate-by-gate algorithm.
As noted in the main text, this algorithm can be applied to any adaptive quantum circuit. The considered MBQC can be viewed as an adaptive quantum circuit 
$U=U_1\otimes \cdots \otimes U_n$ composed of $n$ single-qubit gates. However, since the initial state of this circuit is $\psi_G$ rather than $0^n$, two modifications of the algorithm are needed.
First, the initialization step $x\gets 0^n$ should be replaced by simulating measurement of $\psi_G$, that is, sampling $x$ from the distribution $|\la x|\psi_G\ra|^2$.
As noted above, this requires runtime $O(n)$.
Secondly, the initial state $0^n$ in the definition
of probabilities $P_t(x)$ should be replaced by the surface code state $\psi_G$.
Accordingly, we define
\be
\label{appPt}
P_t(x) = |\la x|U_1 \otimes U_2(x_1) \otimes \cdots \otimes U_t(x_1,\ldots,x_{t-1}) \otimes \underbrace{I \otimes \cdots \otimes I}_{n-t}|\psi_G\ra|^2,
\ee
where $x\in \{0,1\}^n$.
One can easily check that $P_t(x)$ is a normalized probability distribution.
Our goal is to sample $x$ from the final distribution $P_n(x)$.
The modified version of the gate-by-gate algorithm is
\begin{center}
\begin{algorithm}[H]
	\caption{Simulate MBQC with the surface code state $\psi_G$\label{samplingMBQC}}
	\begin{algorithmic}[1]
	\State{Sample $x$ from $P_0(x)=|\la x|\psi_G\ra|^2$}
			\For{$t=1$ to $n$}
			\State{$S\gets \{x,x\oplus e^t\}$}
		\State{Sample $x\in S$ from the probability distribution $P_t(x)/\sum_{y\in S} P_t(y)$}
			      \EndFor
		\State{\textbf{return} $x$}
		\end{algorithmic}
\end{algorithm}
\end{center}
Here $e^t$ is the $n$-bit string with  a single `$1$' at the $t$-th position. The analysis of the gate-by-gate algorithm given in the main text applies almost verbatim to Algorithm~\ref{samplingMBQC}.
Indeed, consider the $t$-th iteration of the $\mathbf{for}$ loop
and let $A=\{1,2,\ldots,n\}\setminus t$. 
We need to check that $P_{t-1}(x_A)=P_t(x_A)$ for all $x\in \{0,1\}^n$, where $x_A$ is the restriction of $x$ onto $A$.
This is equivalent to the identity
\[
\sum_{x_t=0,1} P_{t-1}(x) = \sum_{x_t=0,1} P_t(x).
\]
Since the bits $x_1,\ldots,x_{t-1}$ are fixed, the gate 
$U_t(x_1,\ldots,x_{t-1})$
can be considered as a regular (non-adaptive) single-qubit gate.
The above identity then follows from the unitarity of $U_t$.
The rest of the analysis presented in the main text is unchanged. 

To implement Algorithm~\ref{samplingMBQC}
it suffices to give an efficient subroutine for computing 
the probabilities $P_t(x)$.
Clearly, $P_t(x)$ coincides with the overlap between $\psi_G$ and a tensor product of single-qubit states
$|\phi_j\ra=U_j^\dag(x_1,\ldots,x_{j-1}) |x_j\ra$ or $|\phi_j\ra=|x_j\ra$.
This leads to the following problem.
\begin{problem}[\bf Surface Code Amplitude]
\label{problem:SCA}
Given a planar graph $G=(V,E)$ and a single-qubit state $|\phi_j\ra \in \CC^2$ for every edge $j\in E$.
Compute the overlap $|\la \Phi|\psi_G\ra|^2$, where
$|\Phi\ra = \bigotimes_{j\in E} |\phi_j\ra$.
\end{problem}
As shown in Ref.~\cite{bravyi2007measurement},
Problem~\ref{problem:SCA} can be solved in time $O(n^3)$.
The algorithm of Ref.~\cite{bravyi2007measurement} reduces Problem~\ref{problem:SCA}
to computing the partition function of the Ising model (possibly with complex Boltzmann weights)
defined on the dual graph $G^*$.  The latter is computed 
using the Pfaffian method which goes back to the
seminal works by Kasteleyn~\cite{kasteleyn1961statistics} and
Barahona~\cite{barahona1982computational}.
We conclude that Algorithm~\ref{samplingMBQC}
can simulate MBQC with the surface code state
for any planar graph
in time $O(n^4 T)$, where $T$ is the maximum runtime required
to compute any function $U_j(x_1,\ldots,x_{j-1})$.
The  runtime of Algorithm~\ref{samplingMBQC} can be slightly improved in the special case when $G$ is the 2D square lattice.
In this case Problem~\ref{problem:SCA} can be solved in time $O(n^2)$
using the Majorana fermion representation of the surface code~\cite{bravyi2018correcting}.

Applying the qubit-by-qubit algorithm to simulate 
the considered MBQC requires a subroutine
for computing marginal probabilities 
$\pi_t(y)
=\la \psi| ( |y\ra\la y| \otimes I_{n-t})|\psi\ra$
for all $t=1,2,\ldots,n$.
Let $M=\{1,2,\ldots,t\}$.
Clearly, $\pi_t(y)$ coincides with the overlap between the reduced density matrix
$\rho_M = \mathrm{Tr}_{j\notin M} |\psi_G\ra\la \psi_G|$ and a tensor product of single-qubit
states $\phi_j$ associated with qubits $j\in M$. This leads to the following problem.
\begin{problem}[\bf Surface Code Marginal]
\label{problem:SCM}
Given a planar graph $G=(V,E)$, a subset of edges $M\subseteq E$, and 
a single-qubit state $|\phi_j\ra \in \CC^2$ for every edge $j\in M$. 
Compute the overlap between the reduced density matrix
$\rho_M = \mathrm{Tr}_{j\notin M} |\psi_G\ra\la \psi_G|$ and 
the tensor product state $|\Phi\ra = \bigotimes_{j\in M} |\phi_j\ra$.
In other words, one has to compute the quantity  
\be
\label{overlap_mu}
\mu(G,M,\Phi) := \la \Phi|\rho_M|\Phi\ra 
\ee
\end{problem}

Let us say that a subset of edges $M\subseteq E$ is connected 
if the subgraph of $G$ that includes all edges in $M$ and their endpoints is connected. 
Ref.~\cite{bravyi2007measurement} showed that Problem~\ref{problem:SCM} can be solved in time 
$O(n^3)$ in the special case when both  subsets of edges $M$ and $E\setminus M$ are connected.
As a consequence, Ref.~\cite{bravyi2007measurement} 
showed that the qubit-by-qubit algorithm can efficiently simulate
a restricted class of MBQC such that 
the subsets of edges $\{1,2,\ldots,t\}$
and $\{t+1,t+2,\ldots,n\}$
are connected for all
$t$. In contrast, the connectivity constraint does not matter
for the gate-by-gate algorithm.

Here we complement the results of~\cite{bravyi2007measurement} by showing
that Problem~\ref{problem:SCM} can be $\#P$-hard if the connectivity constraint is removed. 
In other words, the problem of computing marginal probabilities of the state $\psi$ defined in Eq.~(\ref{surface_code_rotated}) is $\#P$-hard in the worst case.

This hardness result prevents one from
applying the qubit-by-qubit algorithm to simulate 
MBQC that do not obey the connectivity constraint.
Indeed, suppose the overlap $\mu(G,M,\Phi)$ is $\#P$-hard to compute
for some planar graph $G$, subset of edges $M\subseteq E$, and  single-qubit
states $\{\phi_i\}_{i\in M}$. 
Order the edges of $G$ such that $M=\{1,2,\ldots,t\}$,
where $t=|M|$. Choose the single-qubit unitaries in MBQC such that 
$|\phi_j\ra=U_j^\dag(0^{j-1}) |0\ra$ for all $j\in M$.
Then $\mu(G,M,\Phi)=\pi_t(0^t)$, that is, the hard-to-compute overlap
coincides with the marginal probability of measuring all-zeros
on the first $t$ qubits. 
The probability that the qubit-by-qubit algorithm samples 
bits $x_1=x_2=\ldots=x_{t-1}=0$ in the first $t-1$ iterations is $\pi_{t-1}(0^{t-1})\ge \pi_t(0^t)$.
Whenever this happens, the algorithm has to compute the 
$\#P$-hard quantity $\pi_t(0^t)$ at the $t$-th iteration.
Our hard instances of Problem~\ref{problem:SCM} 
always result in the positive overlap $\mu(G,M,\Phi)$.
Accordingly, the qubit-by-qubit algorithm
has to solve a $\#P$-hard  problem with a non-zero probability
(although this probability may be exponentially small). 
Here it is essential that MBQC
enforces a particular order in which qubits are measured. A regular (non-adaptive) measurement of the state $\psi$ defined in Eq.~(\ref{surface_code_rotated}) can be simulated efficiently using both gate-by-gate and qubit-by-qubit
algorithms. Indeed, assuming that the graph $G$ is connected, one can always measure qubits in the order that obeys the connectivity constraint and use
the algorithm of~\cite{bravyi2007measurement}.

In the rest of this section we prove $\#P$-hardness
of Problem~\ref{problem:SCM}.
Recall that a subset of edges $M\subseteq E$ is called a perfect matching if
every vertex of the graph has exactly one incident edge from $M$. 
Our starting point is the following hardness result. 
\begin{theorem}[\bf Dagum and Luby~\cite{dagum1992approximating}]
\label{thm:matching}
Exact counting of perfect matchings in a $3$-regular graph is $\#P$-hard.
\end{theorem}
Although this not necessary for our purposes, we note that 
Theorem~\ref{thm:matching}  holds even for a restricted family of $3$-regular graphs 
which are bipartite and 
whose edge set can be represented as a union of three edge-disjoint perfect matchings.
However, this hardness result requires non-planar graphs since
the number of perfect matchings in any planar  graph can be computed efficiently~\cite{kasteleyn1961statistics}.
Here we reduce the $\#P$-hard counting problem considered by Dagum and Luby
to the Surface Code Marginal problem which proves that the latter
is also $\#P$-hard.
\begin{theorem}
There is a polynomial time algorithm that takes as input 
a $3$-regular graph $G'$ and outputs an
instance $(G,M,\Phi)$ of the 
Surface Code Marginal problem 
and a real number $C$  such that the number of perfect matchings
in $G'$ coincides with $C\mu(G,M,\Phi)$.
The size of $G$ is at most polynomial in the size of $G'$.
\end{theorem}
\begin{proof}
First let us introduce some notations.
Suppose $\Theta$ is a graph with a set of vertices $V(\Theta)$ and a set
of $n$ edges $E(\Theta)$. Edges of $\Theta$ are labeled by integers $1,2,\ldots,n$.
Subsets of edges $x\subseteq E(\Theta)$ are identified with $n$-bit strings.
We shall assume that each vertex
$u\in V(\Theta)$
is equipped with a linear order on the set of edges incident to $u$.
Given a bit string $x\in \{0,1\}^n$ (or, equivalently, a subset of edges of $\Theta$),
we shall write $\delta_u(x)$ for the restriction of $x$ onto the set of edges incident to $u$.
In other words, if the ordered set of edges incident to $u$ is $(j,k,\ldots,\ell)$
then $\delta_u(x)=(x_j,x_k,\ldots,x_\ell)$.
If $\Theta$ is a planar graph, the 
order of edges incident to any vertex $u$ must agree with the order in which the edges appear as one circumnavigates $u$ clockwise.

Let $G'=(V',E')$ be the input $3$-regular graph
and $\mathsf{PerfMatch}(G')$ be the number of perfect matchings in $G'$.
If $x\subseteq E'$ is a perfect matching then $y=E'\setminus x$
is a cycle in $G'$ such that each vertex of $G'$ has exactly two incident
edges from $y$ (such cycles are known as $2$-factors).
This shows that
\be
\label{PM1}
\mathsf{PerfMatch}(G') = \sum_{y\in \calZ(G')} \; \prod_{u\in V'} W(\delta_u(y)),
\ee
where $\calZ(G')$ is the cycle space of $G'$ and 
$W\, : \, \{0,1\}^3 \to \RR_+$ is a weight function 
such that
\be
\label{Wfactor}
W(000)=0 \quad \mbox{and} \quad 
W(110)=W(101)=W(011) =1.
\ee
The  function $W$ may take arbitrary values on odd-weight inputs
since  such  inputs  never appear in the sum Eq.~(\ref{PM1}).
Note that the order of edges incident to each vertex does not matter here since $W$ is a symmetric function.

Consider a planar drawing of the input graph $G'$ such that some pairs of edges
of $G'$ may cross. Whenever some pair of edges cross, one of them traverses the crossing
point above the plane and the other traverses below the plane.
Let $G''=(V'',E'')$ be a planar graph obtained by replacing each crossing point in this
planar drawing by a degree-$4$ vertex. 
Thus $V''=V'\cup X$, where $X$ is the set of crossing points.
We claim that
\be
\label{PM2}
\mathsf{PerfMatch}(G') = \sum_{y\in \calZ(G'')} \; \prod_{u\in V'} W(\delta_u(y))
\prod_{u\in X} \tilde{W}(\delta_u(y)),
\ee
where $\tilde{W}\, : \, \{0,1\}^4 \to \RR_+$ is a weight function such that 
\begin{align}
\tilde{W}(z)&=1 \quad \mbox{if} \quad z\in \{0000,1010,0101,1111\},\label{Wcross1}\\
\tilde{W}(z)&=0 \quad \mbox{if} \quad z\in \{1100,0011,1001,0110\}.\label{Wcross2}
\end{align}
Indeed, consider a vertex $u\in X$ and let $(e^1,e^2,e^3,e^4)$ be the
ordered set of edges of $G''$ incident to $u$. By our choice of the linear order,
$e^1,e^3$ originate from an edge of $G'$ that traverses $u$ above the plane
while $e^2,e^4$ originate from an edge of $G'$ that traverses $u$ below the plane
(or vice verse). Thus $\tilde{W}(\delta_u(y))=1$ iff $y_{e^1}=y_{e^3}$ and
$y_{e^2}=y_{e^4}$. In other words, $\tilde{W}(\delta_u(y))=1$ for all $u\in X$
iff $y$ represents some subset of edges in the original graph $G'$.
Combining these observations and Eq.~(\ref{PM1}) one arrives at Eq.~(\ref{PM2}).

The desired planar graph $G$ is obtained from $G''$ by replacing each
vertex $u\in V''$ with a suitable {\em gadget} --
a planar graph $G_u$ such that edges of $G''$
incident to $u$ are identified with dangling edges of $G_u$ lying on its outer face. 
To describe this formally we need some more notations.
Let $\Theta$ be a weighted planar graph with a set of $n$ edges $E(\Theta)$
labeled by integers $1,2,\ldots,n$ and a weight function 
$f\, : \, E(\Theta) \to \CC$.
Write $E(\Theta)=D(\Theta) \cup I(\Theta)$, where
$D(\Theta)$ includes all dangling edges 
(that is, edges that have only one endpoint)
and $I(\Theta)$ includes all
internal edges (that is, edges that have two endpoints). 
The dangling edges are allowed to appear only on the outer face of $\Theta$.
Below we assume that $f(j)=1$ for all $j\in D(\Theta)$ so that non-trivial weights can be assigned only to internal edges.
Also we assume that the set $D(\Theta)$ is equipped with a linear order which agrees with the order in which
the dangling edges  appear as one circumnavigates the outer face of $\Theta$ clockwise.
Given a cycle $x\in \calZ(\Theta)$,
let $\Delta(x)$ be the restriction of $x$ onto the set
$D(\Theta)$. In other words, if the ordered set of dangling edges is $D(\Theta)=(j,k,\ldots,\ell)$
then $\Delta(x)=(x_j,x_k,\ldots,x_\ell)$.
Given a bit string $z$ of length $|D(\Theta)|$, 
define a weighted cycle sum
\be
\label{weighted_sum}
\mathsf{Cycle}(\Theta,z)=\sum_{\substack{x\in \calZ(\Theta)\\  \Delta(x)=z\\}} \; \prod_{j\in x} f(j).
\ee
Each term in the sum is associated with a cycle $x$ in the graph $\Theta$
such that the restriction of $x$
onto the dangling edges $D(\Theta)$ coincides with $z$.
\begin{dfn}
A function $W\, : \, \{0,1\}^k \to \RR_+$ is called admissible 
if there exists a weighted planar graph $\Theta$ with $k$ dangling edges and a normalizing coefficient $\tau>0$ such that 
\be
\label{admissible}
W(z)=\tau \left| \mathsf{Cycle}(\Theta,z)\right|^2
\ee
for all even-weight bit strings $z\in \{0,1\}^k$.
The graph $\Theta$ is called a gadget realizing $W$. 
\end{dfn}
Below we prove the following.
\begin{lemma}[\bf $2$-factor gadget]
\label{lemma:2factor}
The weight function $W$ defined in Eq.~(\ref{Wfactor}) is admissible.
\end{lemma}
\begin{lemma}[\bf Crossing gadget]
\label{lemma:crossing}
The weight function $\tilde{W}$ defined in Eqs.~(\ref{Wcross1},\ref{Wcross2}) is admissible.
\end{lemma}
Let $\Theta$ and $\Gamma$ be the gadgets realizing the weight functions
$W$ and $\tilde{W}$ respectively. 
For each vertex $u\in V''$ define a gadget graph 
\be
G_u = \left\{ \ba{rcl}
\Theta &\mbox{if}& u\in V',\\
\Gamma &\mbox{if}& u\in X.\\
\ea
\right.
\ee
Expressing the weight functions $W(z)$ and $\tilde{W}(z)$
in Eq.~(\ref{PM2}) in terms of the corresponding gadgets
 one gets
\be
\label{PM3}
\mathsf{PerfMatch}(G') =\sigma \sum_{y\in \calZ(G'')} \; 
\prod_{u\in V''} \left|  \mathsf{Cycle}(G_u, \delta_u(y)) \right|^2,
\ee
where $\sigma$ is a product of all normalizing coefficients $\tau$
introduced by the gadgets. 
Let $G=(V,E)$ be a planar graph obtained from $G''$ by
replacing each vertex $u\in V''$ with the gadget graph $G_u$
such that the dangling edges $D(G_u)$ are identified with the 
edges of $G''$ incident to $u$.
We create a new copy of $G_u$ for each $u\in V''$.
We say that an edge $j\in E$  is {\em internal} if it is an internal edge
of some gadget graph $G_u$. 
Otherwise we say that $j$ is an {\em external edge}.
By construction, any external edge $j=(u,v)$
is obtained by identifying a dangling edge of $G_u$
and a dangling edge of $G_v$ for some $(u,v)\in E''$.
Thus the set of external edges of $G$ can be identified with $E''$.
Let
\be
M = \bigcup_{u\in V''} I(G_u)
\ee
be the set of internal edges of $G$.
For every internal edge $j\in I(G_u)$ define a single-qubit (unnormalized) state 
\[
|\phi_j\ra = |0\ra + f_u(j)|1\ra,
\]
where $f_u\, : \, E(G_u)\to \CC$ is the weight function associated with the gadget graph $G_u$.
At this point, we defined a single-qubit state $\phi_j$ for every internal edge $j\in M$.
Let $|\Phi\ra=\bigotimes_{j\in M} |\phi_j\ra$.
We claim that 
\be
\label{PM4}
\mathsf{PerfMatch}(G') = \sigma\cdot |\calZ(G)|\cdot \la \Phi|\rho_M|\Phi\ra,
\ee
where $\rho_M = \mathrm{Tr}_{j\notin M} |\psi_G\ra\la \psi_G|$.
Indeed, it follows directly from the above definitions
that 
\be
\label{cycle_eq1}
 \mathsf{Cycle}(G_u,z) = \la \bigotimes_{j \in I(G_u)} \phi_j|\psi(u,z)\ra
\ee
where $\psi(u,z)$ is a state of $|I(G_u)|$ qubits defined as 
\be
|\psi(u,z)\ra = \sum_{\substack{x\in \calZ(G_u)\\ \Delta_u(x)=z}} |x\cap I(G_u)\ra.
\ee
Here $\Delta_u(x)$ is the restriction of $x$ onto $D(G_u)$.
Furthermore, the surface code state on the graph $G$ can be written as
\be
\label{cycle_eq2}
|\psi_G\ra = \frac1{\sqrt{|\calZ(G)|}} \sum_{y \in \calZ(G'')}
|y\ra_{E''} \bigotimes_{u\in V''} 
|\psi(u,\delta_u(y))\ra_{I(G_u)}
\ee
Here the subscript of a quantum state indicates the subset of qubits that supports this state.
From Eq.~(\ref{cycle_eq2}) one gets
\be
\label{cycle_eq3}
\rho_M = \frac1{|\calZ(G)|}
 \sum_{y \in \calZ(G'')}
 \bigotimes_{u\in V''} 
 |\psi(u,\delta_u(y))\ra\la \psi(u,\delta_u(y))|_{I(G_u)}.
\ee
Combining Eqs.~(\ref{PM3},\ref{cycle_eq1},\ref{cycle_eq3}) proves Eq.~(\ref{PM4}).
If necessary, the single-qubit states $\phi_j$ 
in Eq.~(\ref{PM4}) can be normalized by properly updaing
the coefficient $\sigma$. 
This gives the desired expression 
$\mathsf{PerfMatch}(G')=C\cdot \la \Phi|\rho_M|\Phi\ra$
with $C=\sigma\cdot |\calZ(G)|$.

It remains to prove Lemmas~\ref{lemma:2factor},\ref{lemma:crossing}.

\begin{proof}[\bf Proof of Lemma~\ref{lemma:2factor}]
Consider a weighted planar graph $\Theta$ shown on Figure~\ref{fig:2factor}.
The edge weights of $\Theta$ take values $a=e^{i\pi/3}$ and $b=3^{-1/4}$,
as shown on Figure~\ref{fig:2factor}.
Recall that all dangling edges have weight $1$.
A simple calculation gives
\be
\mathrm{Cycle}(\Theta,000)=1+a^3=0,
\ee
\be
\mathrm{Cycle}(\Theta,011)=\mathrm{Cycle}(\Theta,101)=\mathrm{Cycle}(\Theta,110)=b^2(a+a^2).
\ee
Since $|b^2(a+a^2)|=1$, 
one gets $W(z) = \left| \mathrm{Cycle}(\Theta,z)\right|^2$ for all even-weight strings $z$. 
\begin{figure}[ht]
\centerline{\includegraphics[width=5cm]{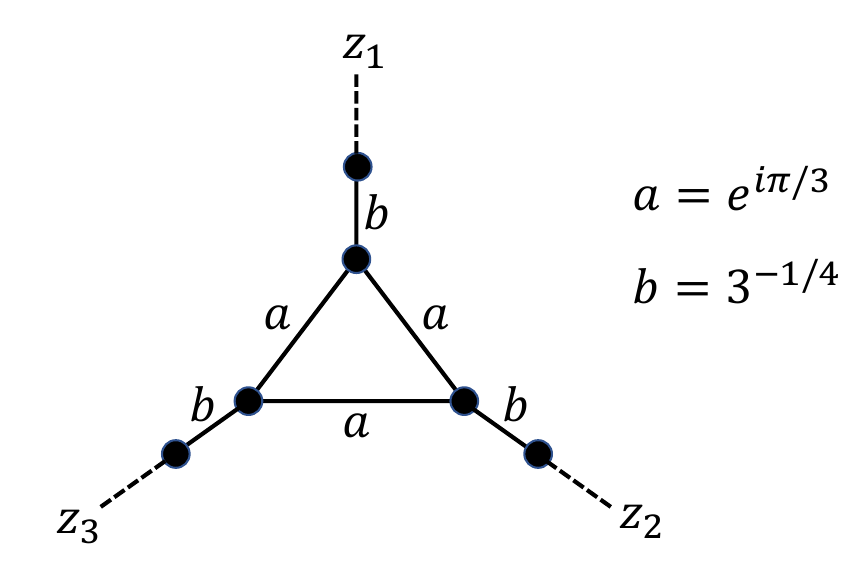}}
\caption{Gadget $\Theta$ implementing a function 
$W(z_1z_2z_3)$ such that 
$W(000)=0$ and $W(011)=W(101)=W(110)=1$.
Dangling and internal edges are shown by dashed and solid lines respectively. Constants $a$ and $b$ are the edge weights of $\Theta$. All dangling edges have weight $1$.
}
\label{fig:2factor}
\end{figure}
\end{proof}
\begin{figure}[ht]
\centerline{\includegraphics[width=8cm]{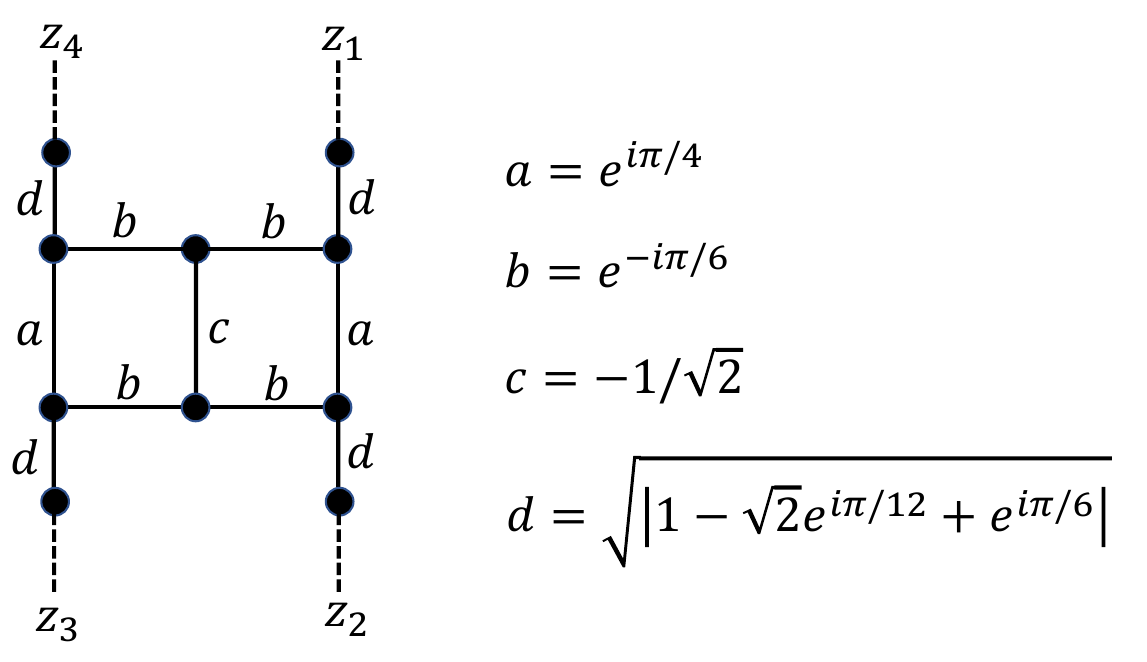}}
\caption{Gadget $\Gamma$ implementing a function $\tilde{W}(z)$
defined in Eqs.~(\ref{Wcross1},\ref{Wcross2}).
Recall that $\tilde{W}(z)=1$ if $z_1=z_3$ and $z_2=z_4$.
Otherwise $\tilde{W}(z)=0$. The function is undefined for
for odd-weight strings $z$.
}
\label{fig:crossing}
\end{figure}
\begin{proof}[\bf Proof of Lemma~\ref{lemma:crossing}]
Consider a weighted planar graph $\Gamma$ shown on Figure~\ref{fig:crossing}.
A simple calculation gives
\begin{align*}
\mathrm{Cycle}(\Gamma,0000)&=1+2ab^2c + a^2 b^4,\\
\mathrm{Cycle}(\Gamma,1100)&=d^2(a + b^2 c + a b^4 + a^2 b^2 c),\\
\mathrm{Cycle}(\Gamma,1010)&=d^2(b^2 c + 2ab^2 + a^2 b^2 c),\\
\mathrm{Cycle}(\Gamma,0110)&=d^2(b^2 + a^2 b^2  + 2ab^2c),\\
\mathrm{Cycle}(\Gamma,1111)&=d^4(a^2 + b^4 + 2a b^2 c).\\
\end{align*}
The symmetries of the gadget imply that 
\[
\mathrm{Cycle}(\Gamma,z_1z_2z_3 z_4) = \mathrm{Cycle}(\Gamma,z_4z_3z_2 z_1)=
\mathrm{Cycle}(\Gamma,z_2z_1z_4 z_3).
\]
By definition, implementing the 
desired function $\tilde{W}(z)$  with a normalization $\tau$ requires
$\tilde{W}(z) = \tau \left| \mathrm{Cycle}(\Gamma,z)\right|^2$ for all even-weight strings $z$.
This is equivalent to
\be
\label{condition1}
|\mathrm{Cycle}(\Gamma,0000)|^2 = |\mathrm{Cycle}(\Gamma,1010)|^2 = |\mathrm{Cycle}(\Gamma,1111)|^2 = \frac1{\tau}
\ee
and
\be
\label{condition2}
\mathrm{Cycle}(\Gamma,1100)=\mathrm{Cycle}(\Gamma,0110)=0.
\ee
One can easily check that these conditions are satisfied for 
\be
a = e^{i\pi/4}, \quad b=e^{-i\pi/6}, \quad c= -1/\sqrt{2}, \quad
d =\sqrt{\left| 1 - \sqrt{2} e^{i\pi/12} + e^{i\pi/6} \right|}
\ee
and
\be
\tau = \frac1{|1+2ab^2c + a^2 b^4|^2} \approx 3.732.
\ee
\end{proof}
\end{proof}

\section{Ground state sensitivity for stoquastic Hamiltonians}

Suppose $H$ is stoquastic, that is, $\la y|H|x\ra\le 0$ for all $x\ne y$. Let $E_0$ be the ground energy of $H$. Suppose $\psi$ is a ground state such that $H|\psi\ra = E_0|\psi\ra$. By Perron–Frobenius theorem, we can assume wlog that $\psi$ has real non-negative amplitudes, i.e. $\la x|\psi\ra\ge 0$ for all $x$.
Consider any string $y$ such that $\la y|\psi\ra>0$.
From $\la y|H|\psi\ra=E_0\la y|\psi\ra$ one gets
\[
(\la y|H|y\ra - E_0)\la y|\psi\ra = \sum_{x\ne y}
|\la y|H|x\ra \la x|\psi\ra |\ge |\la y|H|x\ra \la x|\psi\ra |
\]
for any $x\ne y$. Thus
\[
\frac{|\la y|H|x\ra \la x|\psi\ra|}{|\la y|\psi\ra|}\le \la y|H|y\ra-E_0.
\]
By definition of the sensitivity parameter Eq.~(\ref{sensitivity}),
this implies $s\le \max_y \la y|H|y\ra -E_0$.

\section{Magic ratio Hamiltonians}
We consider a special family of Hamiltonians which can be viewed as a generalization of stoquastic frustration-free Hamiltonians. Let us consider an $n$-qubit Hamiltonian
$$H=-\sum_{a=1}^m P_a$$
where each $P_a$ has the form
$$P_a=\sum_{j}\ket{\phi_{a,j}}\bra{\phi_{a,j}}$$
for some rank-$1$ projectors $\ket{\phi_{a,j}}\bra{\phi_{a,j}}$. For every $n$-qubit state $\ket{\phi}$, define the support of $\ket{\phi}$ as
$$\text{Supp}(\ket{\phi})=\{x\in\{0,1\}^n:\braket{x}{\phi}\neq 0\}.$$
We require an additional key property that for every $a$, the $\ket{\phi_{a,j}}$'s are supported on mutually disjoint (w.r.t $j$) subsets of the computational basis vectors. More formally, we assume that for every $a$,
$$\text{Supp}(\ket{\phi_{a,j}})\cap\text{Supp}(\ket{\phi_{a,j'}})=\emptyset$$
for every $j\neq j'$. We further assume that $H$ is frustration-free, that is, $\psi$ is a ground state of $H$ iff $P_a|\psi\ra=|\psi\ra$
for all $a=1,2,\ldots,m$.
This family of Hamiltonians satisfies the following very special properties.
\begin{lemma}[The magic ratio property]
\label{lemma:magic_ratio}
For every $x,y\in\{0,1\}^n$ and every ground state $\ket{\psi}$ of $H$, if $\bra{x}H\ket{y}\neq 0$, then
\begin{enumerate}[label=(\roman*)]
\item  either $\braket{x}{\psi}=0$ and $\braket{y}{\psi}=0$ or $\braket{x}{\psi}\neq 0$ and $\braket{y}{\psi}\neq 0$;
\item if $\braket{x}{\psi}\neq 0$, then there exists an $a$ and a unique $j$ w.r.t $a$ such that $\frac{\braket{y}{\psi}}{\braket{x}{\psi}}=\frac{\braket{y}{\phi_{a,j}}}{\braket{x}{\phi_{a,j}}}$.
\end{enumerate}
\end{lemma}
\begin{proof}
Let $\ket{\psi}$ be a ground state of $H$, so $P_a\ket{\psi}=\ket{\psi}$ for every $a$ since $H$ is assumed to be frustration-free. Let $x,y\in\{0,1\}^n$ such that $\bra{x}H\ket{y}\neq 0\implies \exists a$ such that $\bra{x}P_a\ket{y}\neq 0\implies \exists$ unique $j$ such that $\bra{x}P_a\ket{y}=\braket{x}{\phi_{a,j}}\braket{\phi_{a,j}}{y}$ with $\braket{x}{\phi_{a,j}}\neq 0$ and $\braket{y}{\phi_{a,j}}\neq 0$. We have that $\braket{y}{\psi}=\bra{y}P_a\ket{\psi}=\braket{y}{\phi_{a,j}}\braket{\phi_{a,j}}{\psi}$ and $\braket{x}{\psi}=\bra{x}P_a\ket{\psi}=\braket{x}{\phi_{a,j}}\braket{\phi_{a,j}}{\psi}$. Thus, $\braket{x}{\psi}=0\iff\braket{\phi_{a,j}}{\psi}=0\iff\braket{y}{\psi}=0$. Furthermore, if $\braket{x}{\psi}\neq 0$, it is clear that $\frac{\braket{y}{\psi}}{\braket{x}{\psi}}=\frac{\braket{y}{\phi_{a,j}}}{\braket{x}{\phi_{a,j}}}$.
\end{proof}
Because of \autoref{lemma:magic_ratio}, we will call $H$ a magic ratio Hamiltonian. Suppose $H$ has a unique ground state $\ket{\psi}$, then we can apply our ground state sampling algorithm to the distribution $\pi(x)=\lvert\braket{x}{\psi}\rvert^2$, $x\in\{0,1\}^n$. \autoref{lemma:magic_ratio} reduces the step of computing the ratio $\frac{\pi(y)}{\pi(x)}$ to locating an $a$ and its uniquely associated $j$ such that $\braket{x}{\phi_{a,j}}\neq 0$ and $\braket{y}{\phi_{a,j}}\neq 0$. This search can be done in $\text{poly}(n,2^k)$ time if each $P_a$ is $k$-local and is specified by a $2^k\times 2^k$ matrix. We remark that stoquastic frustration-free Hamiltonians satisfy properties akin to \autoref{lemma:magic_ratio}, but the key difference is that they possess an additional sign-problem-free characteristic such that their ground states $\ket{\psi'}$ can be chosen so that $\braket{x}{\psi'}\geq 0$ for every $x\in\{0,1\}^n$ \cite{bravyiterhal}. Hence, it is in this regard that $H$ is more general than a stoquastic frustration-free Hamiltonian.

Lemma~\ref{lemma:magic_ratio} provides an upper bound on the sensitivity parameter $s$ defined in Eq.~(\ref{sensitivity})
for any ground state $\psi$ of a frustration-free magic ratio Hamiltonian. Indeed, consider any strings $x\ne y$ such that 
$\la y|H|x\ra\ne 0$ and $\la y|\psi\ra\ne 0$.
Let $M=\{a\, : \, \la y|P_a|x\ra\ne 0\}$.
The proof of Lemma~\ref{lemma:magic_ratio} implies that
for each $a\in M$ one has
\[
\frac{|\la y|P_a |x\ra \la x|\psi\ra|}{|\la y|\psi\ra|}
= |\la x|\phi_{a,j(a)}\ra|^2\le 1.
\]
for some $j(a)$. Thus
\[
s = \frac{|\la y|H|x\ra\la x|\psi\ra |}{|\la y|\psi\ra|}
\le  \sum_{a\in M}\frac{|\la y|P_a|x\ra\la x|\psi\ra |}{|\la y|\psi\ra|} \le |M|\le m.
\]

\end{document}